\documentclass[conference]{IEEEtran}
\usepackage{graphicx}
\usepackage{amsfonts}
\usepackage{amsmath}
\usepackage{amssymb}
\usepackage{mathrsfs} 
\usepackage{color}
\usepackage{bm}
\newtheorem{thm}{Theorem}

\newtheorem{lem}{Lemma}
\newtheorem{proof}{proof}

\newtheorem{exam}{Example}

\ifCLASSINFOpdf \else \fi
\begin{document}

\title{Generalized Reed-Solomon Codes with Sparsest and Balanced
Generator Matrices}

\author{\IEEEauthorblockN{Wentu Song ~ and ~ Kui Cai}
\IEEEauthorblockA{Singapore University of Technology and Design,
Singapore \\ Email: \{wentu\_song, cai\_kui\}@sutd.edu.sg} }
\maketitle

\begin{abstract}
We prove that for any positive integers $n$ and $k$ such that
$n\!\geq\! k\!\geq\! 1$, there exists an $[n,k]$ generalized
Reed-Solomon (GRS) code that has a sparsest and balanced generator
matrix (SBGM) over any finite field of size $q\!\geq\!
n\!+\!\lceil\frac{k(k-1)}{n}\rceil$, where sparsest means that
each row of the generator matrix has the least possible number of
nonzeros, while balanced means that the number of nonzeros in any
two columns differ by at most one. Previous work by Dau et al
(ISIT'13) showed that there always exists an MDS code that has an
SBGM over any finite field of size $q\geq {n-1\choose k-1}$, and
Halbawi et al (ISIT'16, ITW'16) showed that there exists a cyclic
Reed-Solomon code (i.e., $n=q-1$) with an SBGM for any prime power
$q$. Hence, this work extends both of the previous results.
\end{abstract}


\IEEEpeerreviewmaketitle

\section{Introduction}
Maximum distance separable (MDS) codes, especially Reed-Solomon
(RS) codes, with constrained generator matrices are recently
attracting attention for their applications in the scenarios where
encoding is performed in a distributed way
\cite{Hoang13}$-$\cite{SuLi17}. Examples of such scenarios include
wireless sensor networks \cite{Hoang13}, cooperative data exchange
\cite{Yan13,Yan14,SuLi17}, and simple multiple access networks
\cite{Halbawi13,Hoang15}. An interesting problem of this topic is
how to construct MDS codes that have a sparsest and balanced
generator matrix (SBGM), where sparsest means that each row of the
generator matrix has the least possible number of nonzeros, while
balanced means that the number of nonzeros in any two columns
differ by at most one \cite{Hoang13}. More specifically, in an
SBGM of an $[n,k]$ MDS code, the weight of each row is $n-k+1$ and
the weight of each column is either
$\lfloor\frac{k(n-k+1)}{n}\rfloor$ or
$\lceil\frac{k(n-k+1)}{n}\rceil$.

In general, for every MDS code we can easily find a sparsest
generator matrix. The difficulty of this problem is to ensure that
a sparsest generator matrix is also balanced. In \cite{Hoang13},
it was shown that there always exists an MDS code with an SBGM
over any finite field of size $q>{n-1\choose k-1}$ for any $n\geq
k\geq 1$. The authors in \cite{Halbawi16-1} constructed an $[n,k]$
cyclic Reed-Solomon code $($i.e., $n=q-1)$ that has an SBGM for
any prime power $q$ and any $k$ such that $1\leq k\leq n$.
However, it was left as an open problem whether there exists an
$[n,k]_q$ generalized Reed-Solomon (GRS) code with an SBGM for
$k\leq n<q-1$. In this paper, we extends the results in
\cite{Hoang13,Halbawi16-1} by proving that for any positive
integers $n$ and $k$ such that $n\!\geq\! k\!\geq\! 1$, there
exists an $[n,k]$ generalized Reed-Solomon code that has an SBGM
over any finite field $\mathbb F_q$ of size
$q\!\geq\!n+\lceil\frac{k(k-1)}{n}\rceil$.

\subsection{Related Work}
MDS codes with more general constraints on the support of their
generator matrices were studied in \cite{Halbawi13,Yan14,Hoang14}.
A conjecture, called GM-MDS Conjecture, was proposed in
\cite{Hoang14} stating that given any $k\times n$ binary matrix
$M=(m_{i,j})$ that satisfies the so-called MDS Condition, there
exists an $[n,k]_q$ MDS code for any prime power $q\geq n+k-1$
that has a generator matrix $G=(g_{i,j})$ satisfying $g_{i,j}=0$
whenever $m_{i,j}=0$, where the MDS Condition requests that for
any $r\in\{1,2,\cdots,k\}$, the union of the supports of any $r$
rows of $M$ has size at least $n-k+r$.\footnote{Another conjecture
which is equivalent to the GM-MDS Conjecture was proposed in
\cite{Yan14}.} Unfortunately, the GM-MDS Conjecture is proved to
be true only for some very special cases, that is, a) the rows of
$M$ are divided into three groups such that the rows within each
group have the same support \cite{Halbawi13}; or b) the supports
of any two rows of $M$ intersect with at most one element
\cite{Hoang14}; or c) $k\leq 5$ \cite{Heidarzadeh17}.

\section{Preliminary}
For any positive integer $n$, $[n]:=\{1,2,\cdots,n\}$; if $n\leq
0$, $[n]$ is the empty set. For any set $A$, $|A|$ is the size
$($i.e., the number of elements$)$ of $A$. We denote by $\mathbb
F_q$ the field with $q$ elements, where $q\geq 2$ is a prime
power. The support of a row/column vector over $\mathbb F_q$ is
the set of its nonzero coordinates and the weight of a row/column
vector is the size of its support.

A multiset $S$ with underlying set $\{s_1,s_2,\cdots,s_L\}$ is a
set of ordered pairs $S=\{(s_1,n_1),(s_2,n_2),\cdots,(s_L,n_L)\}$,
where each $n_i\geq 0$ is an integer, called the multiplicity of
$s_i$ and denoted by $n_i=\text{mult}_{S}(s_i)$. We also denote
$S$ as
$S=\{s_1,\cdots,s_1,s_2,\cdots,s_2,\cdots,s_L,\cdots,s_L\}$, where
each $s_i$ appears $n_i$ times in $S$ and is also called an
element of $S$. The size $|S|$ of $S$ is the sum of the
multiplicities of its different elements, i.e.,
$|S|=\sum_{i=1}^Ln_i$. Any subset $S_0$ of
$\{s_1,s_2,\cdots,s_L\}$ can be viewed as a multiset such that
$\text{mult}_{S_0}(s_i)=1$ if $s_i\in S_0$, and
$\text{mult}_{S_0}(s_i)=0$ if $s_i\notin S_0$. If
$S'=\{(s_1,m_1),(s_2,m_2),\cdots,(s_L,m_L)\}$ is another multiset,
not necessarily $S'\neq S$, the union of $S$ and $S'$, denoted by
$S\sqcup S'$, is
$\{(s_1,n_1\!+\!m_1),(s_2,n_2\!+\!m_2),\cdots,(s_L,n_L\!+\!m_L)\}$.

Let $\mathcal P_k[x]$ denote the set of polynomials in $\mathbb
F_q[x]$ of degree less than $k$, including the zero polynomial,
where $x$ is an indeterminant. Then $\mathcal P_k[x]$ is a
$k$-dimensional vector space over $\mathbb F_q$ according to the
usual addition and multiplication of polynomials. Let $q\geq n\geq
k$ and $a_1,a_2,\cdots,a_n$ be $n$ distinct elements of $\mathbb
F_q$. The $[n,k]$ generalized Reed-Solomon (GRS) code defined by
$a_1,a_2,\cdots,a_n$ is \cite{MacWilliams}:
$$\mathcal C=\{(f(a_1),f(a_2),\cdots,f(a_n));f\in\mathcal
P_k[x]\}.$$ The code $\mathcal C$ is an MDS code, i.e., the
minimum distance of $\mathcal C$ is $d\!=\!n\!-\!k\!+\!1$. A
generator matrix $G$ of $\mathcal C$ is said to be sparsest and
balanced if $G$ satisfies the following two conditions:
\begin{enumerate}
 \item[(P1)] Sparsest condition: the weight of each row of $G$
 is exactly $n-k+1$;
 \item[(P2)] Balanced condition: the weight of each column of $G$
 is either $\lfloor\frac{k(n-k+1)}{n}\rfloor$ or
 $\lceil\frac{k(n-k+1)}{n}\rceil$.
\end{enumerate}
A GRS code that has a sparsest and balanced generator matrix
(SBGM) is simply called a sparsest and balanced GRS code.

\section{Existence of Sparsest and Balanced GRS Codes}
In this section, we prove that there always exists a sparsest and
balanced $[n,k]$ GRS code for any $n \geq k \geq 1$. Formally, we
have the following theorem.
\begin{thm}\label{main-thm}
For any $n\!\geq\! k\!\geq\! 1$, there exists an $[n,k]$
generalized Reed-Solomon code that has a sparsest and balanced
generator matrix over any field $\mathbb F_q$ of size $q\!\geq\!
n\!+\!\!\lceil\frac{k(k-1)}{n}\rceil$.
\end{thm}

\vspace{0.1cm}
Clearly, $[1,1,\cdots,1]$ is an SBGM of the $[n,1]$ GRS code; and
the identity matrix is an SBGM of the $[n,n]$ GRS code. Hence, in
the following, we only need to consider the case of $$n> k\geq
2.$$

Before proving Theorem \ref{main-thm}, we first prove two lemmas.

First, let $\bm\alpha=(\alpha_1,\alpha_2,\cdots,\alpha_n)$ be an
$n$-tuple of distinct indeterminants. For each subset $Z$ of $[n]$
and $0\leq \ell\leq|Z|$, let $s_Z^{(\ell)}(\bm \alpha)$ be the
$\ell$th elementary symmetric polynomial with respect to
$\{\alpha_j;j\in Z\}$. That is,
$$s_Z^{(0)}(\bm \alpha)=1,$$ and for $1\leq\ell\leq|Z|$,
$$s_Z^{(\ell)}(\bm \alpha)=\sum_{U\subseteq Z\text{~and~}|U|=\ell}
\left(\prod_{j\in U}\alpha_j\right).$$ Then we have the following
lemma.
\begin{lem}\label{zero-subset}
Suppose $n>k\geq 2$. There exists a $k\times n$ binary matrix
$W=(w_{i,j})$ satisfying the following four conditions:
\begin{enumerate}
 \item[(i)] The weight of each row of $W$ is $k-1$;
 \item[(ii)] The weight of each column of $W$ is either
 $\lfloor\frac{k(k-1)}{n}\rfloor$ or
 $\lceil\frac{k(k-1)}{n}\rceil$;
 \vspace{0.1cm}\item[(iii)] $\xi(\bm \alpha)\not\equiv
 0$, where
 \begin{align}\label{def-xi}\hspace{-1.2cm}
 \xi(\bm \alpha)\!=\!\left|\!\!\begin{array}{ccccccccc}
 \vspace{0.2cm}
 s_{Z_1}^{(0)}(\bm \alpha)&s_{Z_2}^{(0)}(\bm \alpha)&\!\cdots\!
 &s_{Z_k}^{(0)}(\bm \alpha)\\
 \vspace{0.2cm}
 s_{Z_1}^{(1)}(\bm \alpha)&s_{Z_2}^{(1)}(\bm \alpha)&\!\cdots\!
 &s_{Z_k}^{(1)}(\bm \alpha)\\
 \vspace{0.2cm}
 s_{Z_1}^{(2)}(\bm \alpha)&s_{Z_2}^{(2)}(\bm \alpha)&\!\cdots\!
 &s_{Z_k}^{(2)}(\bm \alpha)\\
 \vspace{0.1cm}\cdots&\cdots&\!\cdots\!&\cdots\\
 \vspace{0.2cm}
 s_{Z_1}^{(k-1)}(\bm \alpha)&s_{Z_2}^{(k-1)}(\bm \alpha)&\!\cdots\!
 &s_{Z_k}^{(k-1)}(\bm \alpha)\\
 \end{array}\!\!\right|, \end{align} and
 $Z_i$ is the support of the $i$th row of $W$, $\forall i\in[k]$;
 \vspace{0.2cm}
 \item[(iv)] The degree of each $\alpha_i$ in $\xi(\bm \alpha)$ is at
 most $\lceil\frac{k(k-1)}{n}\rceil$.
\end{enumerate}
\end{lem}
\begin{proof}
First, consider $n\geq k(k-1)$. In this case, we have
$\lceil\frac{k(k-1)}{n}\rceil=1$. Let $W=(w_{i,j})$ such that
$w_{i,j}=1$ for each $i\in[k]$ and each $(i-1)(k-1)+1\leq j\leq
i(k-1)$, and $w_{i,j}=0$ otherwise. Then we have
$Z_i=\{j\in[n];(i-1)(k-1)+1\leq j\leq i(k-1)\}$, and
$Z_1,Z_2,\cdots,Z_k$ are mutually disjoint. It is easy to check
that $W$ satisfies conditions (i) $-$ (iv).

In the following, we consider the case that $k(k-1)>n$. Since we
have
assumed $n>k\geq 2$, then we always have
$$k(k-1)>n>k\geq 2.$$
For convenience, we write $k(k-1)$ as
\begin{align}\label{def-r}k(k-1)=an+r \end{align} where $0\leq
r\leq n-1$, and let
\begin{align}\label{def-u}
\bm{\delta}&=(\delta_1, \delta_2, \cdots,
\delta_n)\nonumber\\&=(\overbrace{a+1, \cdots,
a+1}^{r~(a+1)\text{'s}}, ~\overbrace{a, \cdots,
a)}^{n-r~r\text{'s}}.\end{align}

Here we point out some simple facts about $a$ and $\bm \delta$.
First, since $k(k-1)>n>k\geq 2$, if $r=0$, then
\begin{align*}
2\leq a=\left\lfloor\frac{k(k-1)}{n}\right\rfloor
=\left\lceil\frac{k(k-1)}{n}\right\rceil< k-1;\end{align*} if
$0<r\leq n-1$, then \begin{align*}
1\leq a=\left\lfloor\frac{k(k-1)}{n}\right\rfloor<a+1
=\left\lceil\frac{k(k-1)}{n}\right\rceil\leq k-1.\end{align*} So
we always have
\begin{align}\label{a-is-f-c}\delta_j\in\left\{\left\lfloor\frac{k(k-1)}{n}
\right\rfloor,\left\lceil\frac{k(k-1)}{n}\right\rceil\right\}.\end{align}
and
\begin{align}\label{avg}2\leq a+1\leq k-1.\end{align}
Moreover, by \eqref{def-r} and \eqref{def-u}, we have
\begin{align}\label{sum-ones}
\sum_{j=1}^n\delta_j=(a+1)r+a(n-r)=an+r=k(k-1).\end{align}


\textbf{Construction of} $\bm W$: The binary matrix $W$ is
constructed by the following three steps.

\emph{Step 1.} List the elements of the multiset
$$S=\{(1,k-1),(2,k-2),\cdots,(k-1,1)\}$$ in a sequence
\begin{align}\label{def-c}\overline{S}
&=\underbrace{1,2,\cdots,k-1},\underbrace{1,2,\cdots,k-2},
\cdots,\underbrace{1,2},1\nonumber\\&=c_1,c_2,\cdots,c_K\end{align}
where
\begin{align}\label{def-K}
K=\sum_{\ell=1}^{k-1}\ell=\frac{k(k-1)}{2}.\end{align} Then
construct subsets $S_1,S_2,\cdots,S_n$ of $[k]$ by Algorithm 1.

\emph{Step 2.} List the elements of the multiset
$$T=\{(k,k-1),(k-1,k-2),\cdots,(2,1)\}$$ in a sequence
\begin{align}\label{def-e}\overline{T}&=k,\underbrace{k-1,k},
\underbrace{k-2,k-1,k},\cdots,
\underbrace{2,3,\cdots,k}\nonumber\\
&=e_1,e_2,\cdots,e_K\end{align} and let
\begin{align}
\label{def-v}\bm \theta=(\theta_1,\theta_2,\cdots,\theta_n) =\bm
\delta-(|S_1|,|S_2|,\cdots,|S_n|).\end{align} Then construct
subsets $T_1,T_2,\cdots,T_n$ of $[k]$ by Algorithm 2.

\emph{Step 3.} Let $W$ be the $k\times n$ binary matrix such that
for each $j\in[n]$, $Y_j=S_j\cup T_j$ is the support of the $j$th
column of $W$.

\renewcommand\figurename{Fig}
\begin{figure}[htbp]
\begin{center}
\includegraphics[width=8.8cm]{fig1.4}
\end{center}
\end{figure}

\renewcommand\figurename{Fig}
\begin{figure}[htbp]
\begin{center}
\includegraphics[width=8.8cm]{fig1.5}
\end{center}
\end{figure}

Two examples of our construction are given in Section IV.
Moreover, we have the following three claims.

\vspace{0.1cm}\emph{Claim 1}. For each $j\in[n]$, $S_j$ is a
subset of $[k]$ and, when viewed as multisets, we have
$\sqcup_{j=1}^{\lambda_1}S_j=S$, where $\lambda_1$ is the value of
$j$ at the end of the while loop of Algorithm 1.

\vspace{0.1cm}\emph{Claim 2}. For each $j\in[n]$, $T_j$ is a
subset of $[k]$ and $S_j\cap T_j=\emptyset$. Moreover, when viewed
as multisets, we have $\sqcup_{j=1}^nT_j=T$.

\vspace{0.1cm}\emph{Claim 3}. Let
$X^*=\{(1,|S_1|),(2,|S_2|),\cdots,(\lambda_1,|S_{\lambda_1}|)\}$.
Then there exist a unique $\sigma^*\in\mathscr S_k$ and a unique
$(X_1^*,X_2^*,\cdots,X_k^*)\in\mathcal X_{\sigma^*}$ such that
$X^*=X_1^*\sqcup X_2^*\sqcup\cdots\sqcup X_k^*$, where $\mathscr
S_k$ denotes the permutation group on $[k]$ and, for each
$\sigma\in\mathscr S_k$, $\mathcal X_\sigma$ denotes the set of
all tuples $(X_1,X_2,\cdots,X_k)$ such that $X_i\subseteq Z_i$ and
$|X_i|=\sigma(i)-1$, $i=1,2,\cdots,k$.

\vspace{0.1cm} The proof of Claims 1 $-$ 3 are given in Appendices
A $-$ C, respectively.


Note that for each $i\in[k]$, $\text{mult}_{S\sqcup T}(i)=k-1$.
Then by Claims 1 and 2, each $i\in[k]$ is contained by $k-1$ sets
in the collection $\{Y_1,Y_2,\cdots,Y_n\}$, where $Y_j~(j\in[n])$
is the support of the $j$th column of $W$ by our construction. So
each row of $W$ has weight $k-1$, hence condition (i) is
satisfied.

For each $j\in[n]$, by \eqref{a-is-f-c},
$\delta_j\in\{\lfloor\frac{k(k-1)}{n}
\rfloor,\lceil\frac{k(k-1)}{n}\rceil\}.$ So by Claims 1, 2 and
Algorithm 2, the weight of the $j$th column of $W$ is
$|Y_j|=|S_j|+|T_j|=\delta_j\in\{\lfloor\frac{k(k-1)}{n}\rfloor,
\lceil\frac{k(k-1)}{n}\rceil\}$, hence condition (ii) is
satisfied.

For any multiset $X=\{(1,\ell_1),(2,\ell_2),\cdots,(n,\ell_n)\}$,
let
$$\bm\alpha^{X}:=\prod_{j=1}^n\alpha_j^{\ell_j}.$$ Then from \eqref{def-xi},
we have
\begin{align}\label{xi-mult-set}
\xi(\bm \alpha)&=\sum_{\sigma\in\mathscr
S_k}\text{sgn}(\sigma)\prod_{i=1}^ks_{Z_i}^{(\sigma(i)-1)}(\bm
\alpha)\nonumber\\&=\sum_{\sigma\in\mathscr
S_k}\text{sgn}(\sigma)\sum_{(X_1,X_2,\cdots,X_k)\in\mathcal
X_\sigma}\bm\alpha^{X_1}\bm\alpha^{X_2}\cdots\bm\alpha^{X_k}\nonumber\\
&=\sum_{\sigma\in\mathscr
S_k}\text{sgn}(\sigma)\sum_{(X_1,X_2,\cdots,X_k)\in\mathcal
X_\sigma}\bm\alpha^{X_1\sqcup X_2\sqcup\cdots\sqcup X_k}.
\end{align}
where $\text{sgn}(\sigma)$ denotes the sign of the permutation
$\sigma$. By Claim 3, there exist a unique $\sigma^*\in\mathscr
S_k$ and a unique $(X_1^*,X_2^*,\cdots,X_k^*)\in\mathcal
X_{\sigma^*}$ such that $X^*\!=\!X_1^*\sqcup
X_2^*\sqcup\cdots\sqcup X_k^*$. So by \eqref{xi-mult-set},
$\text{sgn}(\sigma^*)\bm\alpha^{X_1^*\sqcup
X_2^*\sqcup\cdots\sqcup X_k^*}$ is a non-zero monomial in $\xi(\bm
\alpha)$. Hence, $\xi(\bm \alpha)\not\equiv 0$ and condition (iii)
is satisfied.

Note that $X_i\subseteq Z_i, ~\forall i\in[k],$ and each column of
$W$ has weight either $\lfloor\frac{k(k-1)}{n} \rfloor$ or
$\lceil\frac{k(k-1)}{n}\rceil$, i.e., each $j\in[n]$ is contained
by at most $\lceil\frac{k(k-1)}{n}\rceil$ sets in
$\{Z_1,Z_2,\cdots,Z_k\}$, where $Z_i$ is the support of the $i$th
row of $W$. So in \eqref{xi-mult-set}, the degree of $\alpha_j$ in
each $\bm\alpha^{X_1\sqcup X_2\sqcup\cdots\sqcup X_k}$ is at most
$\lceil\frac{k(k-1)}{n}\rceil$. Hence, the degree of $\alpha_j$ in
$\xi(\bm\alpha)$ is at most $\lceil\frac{k(k-1)}{n}\rceil$. Hence,
condition (iv) is satisfied, which completes the proof.
\end{proof}

\begin{lem}\label{dgr-nzr}
Suppose $\xi(\alpha_1,\alpha_2,\cdots,\alpha_n)$ is a nonzero
polynomial over the field $\mathbb F_q$ such that the degree of
each $\alpha_i$ is at most $m~(m\geq 1)$. If $q\geq n+m$, then
there exist distinct $a_1,a_2,\cdots,a_n\in \mathbb F_q$ such that
$\xi(a_1,a_2,\cdots,a_n)\neq 0$.
\end{lem}
\begin{proof}
Similar to the Schwartz-Zippel Theorem, this lemma can be proved
by induction on the number of indeterminants $n$. First, for
$n=1$, $\xi(\alpha_1)$ has at most $m$ zeros in $\mathbb F_q$
because the degree of $\alpha_1$ is at most $m$. So there exists
an $a_1\in\mathbb F_q$ such that $\xi(a_1)\neq 0$, provided that
$q\geq 1+m$.

Now assume that $n>1$, $q\geq n+m$ and the induction hypothesis is
true for polynomials of up to $n-1$ indeterminants. Consider the
polynomial $\xi(\alpha_1,\alpha_2,\cdots,\alpha_n)$. Without loss
of generality, assume the degree of $\alpha_1$ in $\xi$ is
$t~(1\leq t\leq m)$. Then we can factor out $\alpha_1$ and obtain
$$\xi(\alpha_1,\alpha_2,\cdots,\alpha_n)
=\sum_{i=0}^t\alpha_1^i\xi_i(\alpha_2,\cdots,\alpha_n),$$ where
$\xi_t(\alpha_2,\cdots,\alpha_n)\not\equiv 0$. Clearly, the degree
of each $\alpha_i~(2\leq i\leq n)$ in $\xi_t$ is at most $m$. The
induction hypothesis implies that there exist distinct
$a_2,\cdots,a_n\in \mathbb F_q$ such that
$\xi_t(a_2,\cdots,a_n)\neq 0$. Then the polynomial
$$\eta(\alpha_1)=\xi(\alpha_1,a_2,\cdots,a_n)
=\sum_{i=0}^t\alpha_1^i\xi_i(a_2,\cdots,a_n)\not\equiv 0$$ and has
degree $t$. Note that $q\geq n+m\geq n+t$. There exists an
$a_1\in\mathbb F_q\backslash\{a_2,\cdots,a_n\}$ such that
$$\xi(a_1,a_2,\cdots,a_n)=\eta(a_1)\neq 0.$$ This completes the
induction.
\end{proof}

\vspace{0.1cm} Now we are able to prove Theorem \ref{main-thm}.
\begin{proof}[Proof of Theorem \ref{main-thm}]
Let $W$ be a $k\times n$ binary matrix satisfying conditions (i)
$-$ (iv) of Lemma \ref{zero-subset}. By Lemma \ref{dgr-nzr}, if
$q\geq n+\lceil\frac{k(k-1)}{n}\rceil$, there exist distinct
$a_1,a_2,\cdots,a_n\in \mathbb F_q$ such that
$\xi(a_1,a_2,\cdots,a_n)\neq 0$.

For each $i\in[k]$, let
\begin{align}\label{def-f_x}
f_i(x)=\prod_{j\in Z_i}(x-a_j)\end{align} where $Z_i$ is the
support of the $i$th row of $W$. Clearly, $f_1(x)$, $f_2(x)$,
$\cdots$, $f_k(x)\in\mathcal P_k[x]$. Moreover, $f_1(x)$,
$f_2(x)$, $\cdots$, $f_k(x)$ are linearly independent in $\mathcal
P_k[x]$, which can be proved as follows. By \eqref{def-f_x}, we
have
\begin{align*}f_i(x)&=\prod_{j\in
Z_i}(x-a_j)\\&=x^{k-1}\!+\!\sum_{\ell=1}^{k-1}\!\left[\sum_{U\subseteq
Z_i,|U|=\ell}\!\left(\prod_{j\in
U}a_j\!\right)\!\right]\!(-1)^\ell x^{k-1-\ell}\\
&=x^{k-1}+\sum_{\ell=1}^{k-1}s_{Z_i}^{(\ell)}(a_1,a_2,\cdots,a_n)(-1)^\ell
x^{k-1-\ell}\end{align*} for each $i\in[k]$. Denote
$c_{i,\ell}:=s_{Z_i}^{(\ell)}(a_1,a_2,\cdots,a_n)$ and
\begin{align*} &C=\\&\!\!\left[\!\!\!\begin{array}{ccccccccc}
1&1&\cdots&1\\
-c_{1,1}&-c_{2,1}&\cdots&-c_{k,1}\\
\cdots&\cdots&\cdots&\cdots\\
(-1)^{k-1}c_{1,k-1}&(-1)^{k-1}c_{2,k-1}&\cdots&(-1)^{k-1}c_{k,k-1}\\
\end{array}\!\!\!\right].
\end{align*}
Then $f_1(x)$,$ f_2(x)$, $\cdots$, $f_k(x)$ are linearly
independent in $\mathcal P_k[x]$ if and only if $\text{det}(C)\neq
0$. From \eqref{def-xi}, we can easily see that
$$\xi(a_1,a_2,\cdots,a_n)=(-1)^{1+2+\cdots+(k-1)}\text{det}(C).$$
Since $\xi(a_1,a_2,\cdots,a_n)\neq 0$, then $\text{det}(C)\neq 0$.
Hence, $f_1(x)$,$ f_2(x),\cdots,f_k(x)$ are linearly independent
in $\mathcal P_k[x]$.

Now, let $\mathcal C$ be the GRS code defined by
$a_1,a_2,\cdots,a_n$ and
\begin{align*} G=\left(\!\begin{array}{ccccccccc}
f_1(a_1)&f_1(a_2)&\cdots&f_1(a_n)\\
f_2(a_1)&f_2(a_2)&\cdots&f_2(a_n)\\
\cdots&\cdots&\cdots&\cdots\\
f_k(a_1)&f_k(a_2)&\cdots&f_k(a_n)\\
\end{array}\!\right).
\end{align*} Since $f_1(x)$,$ f_2(x)$, $\cdots$, $f_k(x)$ are linearly
independent in $\mathcal P_k[x]$, then $G$ is a generator matrix
of $\mathcal C$.

By assumption, $W$ satisfies conditions (i) and (ii) of Lemma
\ref{zero-subset}, that is, the weight of each row of $W$ is $k-1$
and the weight of each column of $W$ is either
$\lfloor\frac{k(k-1)}{n}\rfloor$ or
$\lceil\frac{k(k-1)}{n}\rceil$. Moreover, by \eqref{def-f_x}, for
each $i\in[k]$ and $j\in[n]$, $f_i(a_j)=0$ if and only if $j\in
Z_i$, that is $w_{i,j}=1~($because $Z_i$ is the support of the
$i$th row of $W)$. So according to the construction of $G$, the
number of zeros in every row of $G$ is $k-1$ and the number of
zeros in every column of $G$ is either
$\lfloor\frac{k(k-1)}{n}\rfloor$ or
$\lceil\frac{k(k-1)}{n}\rceil$. Equivalently, the number of ones
in every row of $G$ is $n-k+1$ and the number of ones in every
column of $G$ is either $\lfloor\frac{k(n-k+1)}{n}\rfloor$ or
$\lceil\frac{k(n-k+1)}{n}\rceil$. So $G$ satisfies conditions (P1)
and (P2), hence is an SBGM of $\mathcal C$.
\end{proof}

\section{Examples of the Construction}
As an illustration of our construction, consider the following two
examples, which reflect two typical cases of the output of
Algorithm 1.

\begin{exam}\label{exm-cnstr-W-2}
Let $k=7$ and $n=10$. Then $k(k-1)=4n+2$. So $a=4, r=2$,
$\lfloor\frac{k(k-1)}{n}\rfloor=4$ and
$\lceil\frac{k(k-1)}{n}\rceil=5$. According to \eqref{def-u}, we
have
$$\bm \delta=(5,5,4,4,4,4,4,4,4,4)$$ and according to \eqref{def-c}, we have
$$\overline{S}=\underbrace{1,2,3,4,5,6},\underbrace{1,2,3,4,5},
\underbrace{1,2,3,4},\underbrace{1,2,3},\underbrace{1,2},1.$$ By
Algorithm 1, $\overline{S}$ is divided into $S_1,\cdots,S_6$ as
follows:
\begin{align}\label{exam-eq-2}
\underbrace{1,2,3,4,5}_{S_1},\underbrace{6,1,2,3,4}_{S_2},
\underbrace{5,1,2,3}_{S_3},\underbrace{4,1,2,3}_{S_4},
\underbrace{1,2}_{S_5},\underbrace{1}_{S_6}\end{align} and
$S_7=\cdots=S_{n}=\emptyset$. Hence,
\begin{align*}(|S_1|,|S_2|,\cdots,|S_n|)=(5,5,4,4,2,1,0,0,0,0)\end{align*}
and according to \eqref{def-v}, we have \begin{align*}\bm
\theta&=\bm
\delta-(|S_1|,|S_2|,\cdots,|S_n|)\\&=(0,0,0,0,2,3,4,4,4,4).\end{align*}
Moreover, according to \eqref{def-e}, we have
$$\overline{T}=7,\underbrace{6,7},\underbrace{5,6,7},\underbrace{4,5,6,7},
\underbrace{3,4,5,6,7},\underbrace{2,3,4,5,6,7}.$$ Then by
Algorithm 2, we have $T_1=\cdots=T_4=\emptyset$ and $\overline{T}$
is divided into $T_5,\cdots,T_{10}$ as follows:
\begin{align*}\underbrace{7,6}_{T_5}, \underbrace{7,5,6}_{T_6},
\underbrace{7,4,5,6}_{T_7}, \underbrace{7,3,4,5}_{T_8},
\underbrace{6,7,2,3}_{T_9},
\underbrace{4,5,6,7}_{T_{10}}.\end{align*} So we obtain
\begin{align*}
W=\left[\begin{array}{cccccccccccc}
1&1&1&1&1&1&0&0&0&0\\
1&1&1&1&1&0&0&0&1&0\\
1&1&1&1&0&0&0&1&1&0\\
1&1&0&1&0&0&1&1&0&1\\
1&0&1&0&0&1&1&1&0&1\\
0&1&0&0&1&1&1&0&1&1\\
0&0&0&0&1&1&1&1&1&1\\
\end{array}\right].
\end{align*}

We can easily check that Claims 1 and 2 are true. We now check
that Claim 3 is true. From \eqref{exam-eq-2}, we have
$\lambda_1=6$ and
\begin{align*}X^*&=\{(1,|S_1|),(2,|S_2|),\cdots,(6,|S_{6}|)\}\\
&=\{(1,5),(2,5),(3,4),(4,4),(5,2),(6,1)\}.\end{align*} Suppose
$$X^*=X_1^*\sqcup X_2^*\sqcup\cdots\sqcup X_k^*$$ for some
$\sigma^*\in\mathscr S_k$ and some
$(X_1^*,X_2^*,\cdots,X_k^*)\in\mathcal X_{\sigma^*}$. We show that
$\sigma^*$ and $(X_1^*,X_2^*,\cdots,X_k^*)$ are unique as follows.

First, note that for each $j\in\{1,2,3,4\}$,
$\text{mult}_{X^*}(j)$ equals the weight of the $j$th column of
$W$. Then considering the first four columns of $W$, we have
$\{1,2,3,4\}\subseteq (\cap_{i=1}^3X^*_{i})$, $\{1,2,4\}\subseteq
X^*_{4}$, $\{1,3\}\subseteq X^*_{5}$ and $\{2\}\subseteq X^*_6$.
So it must be that $\sigma^*(7)=1$ and $X^*_7=\emptyset$.
Recursively, we obtain $\sigma^*(6)=2$ and $X^*_6=\{2\}$;
$\sigma^*(5)=3$ and $X^*_5=\{1,3\}$; $\sigma^*(4)=4$ and
$X^*_4=\{1,2,4\}$. And hence, we have $\sigma^*(i)\in\{5,6,7\}$
for each $i\in\{1,2,3\}$, and $\text{mult}_{X_7^*\sqcup
X_6^*\sqcup X_5^*\sqcup X_4^*}(j)=0$ for $j=5,6$.

Further, consider the first five columns of $W$. Since
$\text{mult}_{X_7^*\sqcup X_6^*\sqcup X_5^*\sqcup X_4^*}(5)=0$,
then $\text{mult}_{X_1^*\sqcup X_2^*\sqcup
X_3^*}(5)=\text{mult}_{X^*}(5)=2$ and $\{1,2,3,4,5\}\subseteq
(X^*_{1}\cap X^*_{2})$. So $\sigma^*(3)=5$ and
$X^*_3=\{1,2,3,4\}$. Similarly, considering the first six columns
of $W$, we can obtain $\sigma^*(2)=6$ and $X^*_2=\{1,2,3,4,5\}$.
And finally, we can obtain $\sigma^*(1)=7$ and
$X^*_1=\{1,2,3,4,5,6\}$.

Hence, $\sigma^*\in\mathscr S_k$ and
$(X_1^*,X_2^*,\cdots,X_k^*)\in\mathcal X_{\sigma^*}$ are uniquely
determined. That is, Claim 3 is true.

As discussed in the proof of Lemma \ref{zero-subset}, $W$
satisfies conditions (i) $-$ (iv) of Lemma \ref{zero-subset}.
\end{exam}

\begin{exam}\label{exm-cnstr-W-1}
Let $k=7$ and $n=13$. Then $k(k-1)=3n+3$. So $a=3, r=3$,
$\lfloor\frac{k(k-1)}{n}\rfloor=3$ and
$\lceil\frac{k(k-1)}{n}\rceil=4$. According to \eqref{def-u}, we
have
$$\bm \delta=(4,4,4,3,3,3,3,3,3,3,3,3,3)$$
and according to \eqref{def-c}, we have
$$\overline{S}=\underbrace{1,2,3,4,5,6},\underbrace{1,2,3,4,5},
\underbrace{1,2,3,4},\underbrace{1,2,3},\underbrace{1,2},1.$$ By
Algorithm 1, $\overline{S}$ is divided into $S_1,\cdots,S_7$ as
follows:
\begin{align}\label{exam-eq-1}
\underbrace{1,2,3,4}_{S_1},\underbrace{5,6,1,2}_{S_2},
\underbrace{3,4,5,1}_{S_3},\underbrace{2,3,4}_{S_4},
\underbrace{1,2,3}_{S_5},\underbrace{1,2}_{S_6},
\underbrace{~1~}_{S_7}.\end{align} And $S_8=\cdots=S_n=\emptyset$.
Hence,
\begin{align*}(|S_1|,|S_2|,\cdots,|S_n|)
=(4,4,4,3,3,2,1,0,0,0,0,0,0)\end{align*} and according to
\eqref{def-v},
\begin{align*}\bm \theta&=\bm \delta-(|S_1|,|S_2|,\cdots,|S_n|)\\
&=(0,0,0,0,0,1,2,3,3,3,3,3,3).\end{align*} Moreover, according to
\eqref{def-e}, we have
$$\overline{T}=7,\underbrace{6,7},\underbrace{5,6,7},\underbrace{4,5,6,7},
\underbrace{3,4,5,6,7},\underbrace{2,3,4,5,6,7}.$$ Then by
Algorithm 2, we have $T_1=\cdots=T_5=\emptyset$ and $\overline{T}$
is divided into $T_6,\cdots,T_{13}$ as follows:
\begin{align*}
\underbrace{7}_{T_6},\underbrace{6,7}_{T_7},
\underbrace{5,6,7}_{T_8},\underbrace{4,5,6}_{T_9},
\underbrace{7,3,4}_{T_{10}},\underbrace{5,6,7}_{T_{11}},
\underbrace{2,3,4}_{T_{12}},\underbrace{5,6,7}_{T_{13}}.
\end{align*} So we obtain
\begin{align*}
W=\left[\begin{array}{ccccccccccccc}
1&1&1&0&1&1&1&0&0&0&0&0&0\\
1&1&0&1&1&1&0&0&0&0&0&1&0\\
1&0&1&1&1&0&0&0&0&1&0&1&0\\
1&0&1&1&0&0&0&0&1&1&0&1&0\\
0&1&1&0&0&0&0&1&1&0&1&0&1\\
0&1&0&0&0&0&1&1&1&0&1&0&1\\
0&0&0&0&0&1&1&1&0&1&1&0&1\\
\end{array}\right].
\end{align*}
We can check that Claims 1 and 2 are true. Moreover, let
\begin{align*}X^*&=\{(1,|S_1|),(2,|S_2|),\cdots\!,(7,|S_{7}|)\}\\
&=\{(1,4),(2,4),(3,4),(4,3),(5,3),(6,2),(7,1)\}\end{align*} and
suppose $X^*=X_1^*\sqcup X_2^*\sqcup\cdots\sqcup X_k^*$ for some
$\sigma^*\in\mathscr S_k$ and some
$(X_1^*,X_2^*,\cdots,X_k^*)\in\mathcal X_{\sigma^*}$. Then similar
to Example \ref{exm-cnstr-W-2}, we can obtain $\sigma^*(i)=k-i+1,
~\forall i\in[k]$, and $X^*_7=\emptyset$, $X^*_6=\{2\}$,
$X^*_5=\{2,3\}$, $X^*_4=\{1,3,4\}$, $X^*_3=\{1,3,4,5\}$,
$X^*_2=\{1,2,4,5,6\}$, $X^*_1=\{1,2,3,5,6,7\}$. So both
$\sigma^*\in\mathscr S_k$ and
$(X_1^*,X_2^*,\cdots,X_k^*)\in\mathcal X_{\sigma^*}$ are unique
and Claim 3 is true.
\end{exam}

\section{Conclusion}
We show that for any $n\!\geq\! k\!\geq\! 1$, there exists an
$[n,k]$ sparsest and balanced GRS code over any field $\mathbb
F_q$ with size $q\!\geq\! n+\!\lceil\frac{k(k-1)}{n}\rceil$. It is
still an open problem whether $[n,k]$ sparsest and balanced GRS
codes exist when the field size $q$ satisfies
$n+1<q<n+\lceil\frac{k(k-1)}{n}\rceil$.

\appendices

\section{Proof of Claim 1}
In this appendix, we are to prove Claim 1.

By \eqref{def-u} and \eqref{avg}, we have $\delta_\ell\!\leq\!
a+1\!\leq\! k-1$ for each $\ell\in[n]$. Then there exists a unique
$\lambda_{a+1}\in[n]$ such that
\begin{align}\label{lmd-a-plus-1}
\sum_{\ell=1}^{\lambda_{a+1}-1}\delta_\ell<\sum_{\ell=a+1}^{k-1}\ell\leq
\sum_{\ell=1}^{\lambda_{a+1}}\delta_\ell.\end{align} According to
Algorithm 1, we have
\begin{align}\label{size-S-leq-j0}
|S_j|=\delta_j, ~\forall
j\in\{1,2,\cdots,\lambda_{a+1}\}\end{align} and each of
$S_1,S_2,\cdots,S_{\lambda_{a+1}}$ is a subset of $[k]$. Moreover,
since $\delta_{\lambda_{a+1}}\leq a+1$, then from
\eqref{lmd-a-plus-1}, we obtain
$$\sum_{\ell=1}^{\lambda_{a+1}}|S_\ell|=\sum_{\ell=1}^{\lambda_{a+1}}\delta_\ell=
\left(\sum_{\ell=a+1}^{k-1}\ell\right)+t_0$$ for some
$t_0\!\in\!\{0,1,\cdots,a\}$. We need to consider the following
two cases.

\textbf{Case 1}. $t_0\in\{1,2,\cdots,a\}$.

Then according to Algorithm 1, we have
\begin{itemize}
 \item[$\bullet$] For $j=\lambda_{a+1}+\ell$ and $1\leq\ell \leq a-t_0$,
 \begin{align}\label{elmt-S-1}
 S_{j}&=\{t_0\!+\!1,t_0\!+\!2,\cdots\!,a\!-\!\ell\!+\!1,
 1,2,\cdots\!,t_0\}\nonumber\\
 &=\{1,2,\cdots,a-\ell+1\};\end{align}
 \item[$\bullet$] For $j=\lambda_{a+1}+\ell$ and $a-t_0+1\leq\ell\leq a-1$,
 \begin{align}\label{elmt-S-2}
 S_{j}=\{1,2,\cdots,a-\ell\}\end{align}
\end{itemize}
Moreover, $\lambda_{1}:=\lambda_{a+1}+a-1$ is the value of $j$ at
the end of the while loop of Algorithm 1 and
\begin{align}\label{s-v-S-1}
&(|S_1|,|S_2|,\cdots\!,|S_{\lambda_{1}}|)=\nonumber\\
&(\delta_1,\cdots\!,\delta_{\lambda_{a+1}},a,a\!-\!1,\cdots\!,t_0\!+\!1,t_0\!-\!1,
\cdots\!,2,1).\end{align}

\textbf{Case 2}. $t_0=0$.

Then according to Algorithm 1, we have
\begin{itemize}
 \item [$\bullet$] For $j=\lambda_{a+1}+\ell$ and $\ell\in[a]$,
 \begin{align}\label{elmt-S-3}S_{j}=\{1,2,\cdots,a-\ell+1\}.\end{align}
\end{itemize}
Moreover, $\lambda_{1}:=\lambda_{a+1}+a$ is the value of $j$ at
the end of the while loop of Algorithm 1 and
\begin{align}\label{s-v-S-2}
&(|S_1|,|S_2|,\cdots\!,|S_{\lambda_{1}}|)=\nonumber\\
&(\delta_1,\cdots\!,\delta_{\lambda_{a+1}},a,a\!-\!1,\cdots,2,1).\end{align}

In both cases, clearly, each of
$S_{\lambda_{a+1}+1},\cdots,S_{\lambda_{1}}$ is a subset of $[k]$.
Moreover, we have $\lambda_{1}\leq n$, which can be proved by
contradiction as follows. Suppose $\lambda_{1}>n$. Then we have
\begin{align*}\sum_{j=1}^{\lambda_{1}}|S_j|
+\sum_{\ell=1}^a\ell>\sum_{j=1}^{n}|S_j|
+\sum_{\ell=1}^a\ell.\end{align*} Moreover, since $\delta_j\leq
a+1$ for all $j\in[n]~($see \eqref{def-u}$)$, then from
\eqref{s-v-S-1} and \eqref{s-v-S-2}, we have
\begin{align*}\sum_{j=\lambda_{a+1}+1}^{n}|S_j|
+\sum_{\ell=1}^a\ell\geq\sum_{j=\lambda_{a+1}+1}^{n}\delta_j.\end{align*}
From the above two inequalities, we have
\begin{align}\label{ctrn-sum}\sum_{j=1}^{\lambda_{1}}|S_j|
+\sum_{\ell=1}^a\ell>\sum_{j=1}^{n}\delta_j=k(k-1)\end{align}
where the last equation comes from \eqref{sum-ones}. However,
combining the facts
$\sum_{j=1}^{\lambda_{1}}|S_j|\!=\!K\!=\!\frac{k(k-1)}{2}$ and
$a\!<\!k\!-\!1$, we have
$$\sum_{j=1}^{\lambda_{1}}|S_j|
+\sum_{\ell=1}^a\ell<\frac{k(k-1)}{2}+\sum_{\ell=1}^{k-1}\ell=
k(k-1)$$ which contradicts to \eqref{ctrn-sum}. Hence we proved
that $\lambda_{1}\leq n$.

Further, according to Algorithm 1, we have
\begin{align}\label{C1-S-j1-n}S_{\lambda_{1}+1}=\cdots=S_n=\emptyset.\end{align}
So in Case 1, we have
\begin{align}\label{S-size-v-1}
&(|S_1|,|S_2|,\cdots\!,|S_{n}|)=\nonumber\\
&(\delta_1,\cdots,\delta_{\lambda_{a+1}},a,a\!-\!1\cdots\!,
t_0\!+\!1,t_0\!-\!1,\cdots\!,2,1,
\overbrace{0,\cdots,0}^{n-\lambda_{1}~\text{zeros}})\end{align}
where $t_0\in\{1,2,\cdots,a\}$ and
$\lambda_{1}=\lambda_{a+1}+a-1$; in Case 2, we have
\begin{align}\label{S-size-v-2}
&(|S_1|,|S_2|,\cdots\!,|S_{n}|)=\nonumber\\
&(\delta_1,\cdots,\delta_{\lambda_{a+1}},a,a\!-\!1\cdots\!,2,1,
\overbrace{0,\cdots,0}^{n-\lambda_{1}~\text{zeros}})\end{align}
where $\lambda_{1}=\lambda_{a+1}+a-1$. In both cases, each $S_j$
is a subset of $[k]$ and, as multisets,
$\sqcup_{j=1}^nS_j=\sqcup_{j=1}^{\lambda_{1}}S_j=S$.

In Example \ref{exm-cnstr-W-2}, we have $a+1=5$. From
\eqref{exam-eq-2}, we can obtain $\lambda_6=2$, $\lambda_5=3$ and
$t_0=3$. So this example falls into Case 1 and
$\lambda_1=\lambda_{a+1}+a-1=6$.

In Example \ref{exm-cnstr-W-1}, we have $a+1=4$. From
\eqref{exam-eq-1}, we can obtain $\lambda_6=2$, $\lambda_5=3$,
$\lambda_4=4$ and $t_0=0$. So this example falls into Case 2 and
$\lambda_1=\lambda_{a+1}+a=7$.

\section{Proof of Claim 2}

To prove Claim 2, we continue considering the two cases discussed
in Appendix A.

First, consider Case 1. We need to divide this case into the
following four subcases according to the value of $r$.

\textbf{Case 1.1}: $r\leq \lambda_{a+1}$.

Then by \eqref{def-u}, \eqref{def-v} and \eqref{S-size-v-1}, we
have
\begin{align*}\bm \theta&=\bm \delta-(|S_1|,|S_2|,\cdots,|S_{n}|)\\
&=(\overbrace{0,\cdots,0}^{\lambda_{a+1}~\text{zeros}},
0,1,\cdots,a-t_0-1,a-t_0+1,\nonumber\\&~~~~~\cdots,a-2,a-1,\delta_{\lambda_{1}+1},
\cdots,\delta_n)\end{align*} where
$\delta_{\lambda_{1}+1}=\cdots=\delta_{n}=a$. That is,
\begin{align}\label{vector-v-1-1}
\theta_{j}=\left\{\begin{aligned}
&0,& &1\leq j\leq \lambda_{a+1}; \\
&\ell-1,& &j=\lambda_{a+1}+\ell\text{~and~}1\leq\ell\leq a-t_0;\\
&\ell,& &j=\lambda_{a+1}+\ell\text{~and~}a-t_0+1\leq\ell\leq a-1;\\
&\delta_j=a,& &\lambda_{1}+1\leq j\leq n. \\
\end{aligned} \right.
\end{align}
According to Algorithm 2, we have
\begin{itemize}
 \item[$\bullet$] For $1\leq j\leq\lambda_{a+1}+1$, $T_j=\emptyset;$
 \vspace{0.1cm}
 \item[$\bullet$] For $j=\lambda_{a+1}+\ell$ and $2\leq\ell\leq a-t_0$,
 \begin{align}\label{elmt-T-1}
 T_{j}&=\{k\!-\!(\ell\!-\!1)\!+\!1,k\!-\!(\ell\!-\!1)\!+\!2, \cdots\!,k\}\nonumber
 \\&=\{k\!-\!\ell\!+\!2,k\!-\!\ell\!+\!3,\cdots\!,k\};\end{align}
 \item[$\bullet$] For $j=\lambda_{a+1}+\ell$ and $a-t_0+1\leq\ell\leq a-1$,
 \begin{align}\label{elmt-T-2}
 T_{j}\!&=\!\{k\!-\!(a\!-\!t_0)\!+\!1,\cdots,k,k\!-\!\ell\!+\!1,
 \cdots,k\!-\!(a-\!t_0)\}\nonumber\\&=\!\{k\!-\!\ell\!+\!1,k\!-\!\ell\!+\!2,
 \cdots,k\}\end{align}
 \item[$\bullet$] Finally,
 $T_{\lambda_{1}+1}, T_{\lambda_{1}+2},\cdots\!,T_n$ are obtained by dividing the
 sequence
 $$\underbrace{k\!-\!(a\!-\!t_0)\!+\!1,\cdots\!,k},
 \underbrace{k\!-\!a\!+\!1,
 \cdots\!,k},\cdots\!, \underbrace{2,\cdots\!,k}$$ into $n-\lambda_{1}$
 segments of length $a$, and $T_{\lambda_{1}+j}$ is then formed by the
 elements of the $j$th segment, $1\!\leq\! j\!\leq\! n\!-\!\lambda_{1}$.
\end{itemize}
Clearly, each $T_j$ is a subset of $[k]$. Moreover, since
$$\sum_{j=1}^n\theta_j=\sum_{j=1}^n\delta_j-\sum_{j=1}^n|S_j|=\frac{k(k-1)}{2}=K$$
which is equal to the length of $\overline{T}$. So Algorithm 2
always divides $\overline{T}$ into $T_{1}, T_{2},\cdots,T_n$ of
size $\theta_1,\theta_2,\cdots,\theta_n$, respectively. Hence, as
multisets, we have $\sqcup_{j=1}^nT_j=T$.

Note that $T_j=\emptyset$ for $1\leq j\leq\lambda_{a+1}+1$, and
$S_j=\emptyset$ for $\lambda_{1}+1\leq j\leq n$. So $S_j\cap
T_j=\emptyset$ for
$j\in\{1,\cdots,\lambda_{a+1}+1\}\cup\{\lambda_{1}+1, \cdots,n\}$.
Moreover, for $\lambda_{a+1}+2\leq j\leq \lambda_{1}$, by
\eqref{elmt-S-1}, \eqref{elmt-S-2}, \eqref{elmt-S-3},
\eqref{elmt-T-1} and \eqref{elmt-T-2}, we have $$\max(S_j)\leq
a-\ell+1<k-\ell+1\leq\min(T_j).$$ So $S_j\cap T_j=\emptyset$ for
$j\in\{\lambda_{a+1}+2,\cdots,\lambda_{1}\}$. Hence, we have
$S_j\cap T_j=\emptyset$ for all $j\in[n]$.

\textbf{Case 1.2}: $\lambda_{a+1}<r\leq \lambda_{a+1}+a-t_0$.

Then $r=\lambda_{a+1}+t_1$, where $1\leq t_1\leq a-t_0$, and
\begin{align*}
\theta_{j}=\left\{\begin{aligned}
&0,& &1\leq j\leq \lambda_{a+1}; \\
&\ell,& &j=\lambda_{a+1}+\ell\text{~and~}1\leq\ell\leq t_1;\\
&\ell-1,& &j=\lambda_{a+1}+\ell\text{~and~}t_1+1\leq\ell\leq a-t_0;\\
&\ell,& &j=\lambda_{a+1}+\ell\text{~and~}a-t_0+1\leq\ell\leq a-1;\\
&\delta_j=a,& &\lambda_{1}+1\leq j\leq n. \\
\end{aligned} \right.
\end{align*}
According to Algorithm 2, we have
\begin{itemize}
 \item [$\bullet$] For $1\leq j\leq\lambda_{a+1}$, $T_j=\emptyset;$
 \item [$\bullet$] For $j=\lambda_{a+1}+\ell$ and $1\leq\ell\leq t_1$,
 $$T_{j}=\{k\!-\!\ell\!+\!1,k\!-\!\ell\!+\!2,\cdots\!,k\};$$
 \item [$\bullet$] For $j=\lambda_{a+1}+\ell$ and $t_1+1\leq\ell\leq a-t_0$,
 $$T_{j}\!=\!\{k\!-\!\ell\!+\!1,k\!-\!\ell\!+\!2,\cdots\!,k\}
 \backslash\{k\!-\!l\!+\!t_1\!+\!1\};$$
 \item [$\bullet$] For $j=\lambda_{a+1}+\ell$ and $a-t_0+1\leq\ell\leq a-1$
 $$T_{j}\!=\!\{k\!-\!\ell\!+\!1,k\!-\!\ell\!+\!2,
 \cdots\!,k\};$$
 \item [$\bullet$] Finally, $T_{\lambda_{1}+1}, T_{\lambda_{1}+2},\cdots\!,T_n$
 are obtained by dividing the sequence
 $$\underbrace{k\!-\!(a\!-\!t_0-\!t_1)\!+\!1,\cdots\!,k},\underbrace{k\!-\!a\!+\!1,
 \cdots\!,k},\cdots\!,\underbrace{2,\cdots\!,k}$$ into $n\!\!-\!\!\lambda_{1}$
 segments of length $a$, and $T_{\lambda_{1}+j}$ is formed by the
 elements of the $j$th segment, $1\!\leq\! j\!\leq\! n\!-\!\lambda_{1}$.
\end{itemize}

\textbf{Case 1.3}: $\lambda_{a+1}+a-t_0<r\leq \lambda_{1}$.

Then $r=\lambda_{a+1}+t_2$, where $a-t_0+1\leq t_2\leq a-1$, and
\begin{align*}
\theta_{j}=\left\{\begin{aligned}
&0,& &1\leq j\leq \lambda_{a+1}; \\
&\ell,& &j=\lambda_{a+1}+\ell\text{~and~}1\leq\ell\leq a-t_0;\\
&\ell+1,& &j=\lambda_{a+1}+\ell\text{~and~}a-t_0+1\leq\ell\leq t_2;\\
&\ell,& &j=\lambda_{a+1}+\ell\text{~and~}t_2+1\leq\ell\leq a-1;\\
&\delta_j=a,& &\lambda_{1}+1\leq j\leq n. \\
\end{aligned} \right.
\end{align*}
According to Algorithm 2, we have
\begin{itemize}
 \item [$\bullet$] For $1\leq j\leq\lambda_{a+1}$, $T_j=\emptyset;$
 \item [$\bullet$] For $j=\lambda_{a+1}+\ell$ and $1\leq\ell\leq a-t_0$,
 $$T_{j}=\{k\!-\!\ell\!+\!1,k\!-\!\ell\!+\!2,\cdots\!,k\};$$
 \item [$\bullet$] For $j=\lambda_{a+1}+\ell$ and $a-t_0+1\leq\ell\leq t_2$,
 $$T_{j}\!=\!\{k\!-\!\ell,k\!-\!\ell\!+\!1,\cdots\!,k\};$$
 \item [$\bullet$] For $j=\lambda_{a+1}+\ell$ and $t_2+1\leq\ell\leq a-1$
 $$T_{j}\!=\!\{k\!-\!\ell,k\!-\!\ell\!+\!1, \cdots\!,k\}
 \backslash\{k\!-\!a\!+\!t_0\!-\!\ell\!+\!t_2\};$$
 \item [$\bullet$] Finally, $T_{\lambda_{1}+1}, T_{\lambda_{1}+2},\cdots\!,T_n$
 are obtained by dividing the sequence
 $$\underbrace{k\!-\!2a\!+\!t_0\!+\!t_2\!+\!1,\cdots\!,k},\underbrace{k\!-\!a,
 \cdots\!,k},\cdots,\underbrace{2,\cdots\!,k}$$ into $n\!\!-\!\!\lambda_{1}$
 segments of length $a$, and $T_{\lambda_{1}+j}$ is formed by the
 elements of the $j$th segment, $1\!\leq\! j\!\leq\! n\!-\!\lambda_{1}$.
\end{itemize}

\textbf{Case 1.4}: $\lambda_{1}<r<n$.

Then by \eqref{def-u}, \eqref{def-v} and \eqref{S-size-v-1}, we
have
\begin{align*}
\theta_{j}=\left\{\begin{aligned}
&0,& &1\leq j\leq \lambda_{a+1}; \\
&\ell,& &j=\lambda_{a+1}+\ell\text{~and~}1\leq\ell\leq a-t_0;\\
&\ell+1,& &j=\lambda_{a+1}+\ell\text{~and~}a-t_0+1\leq\ell\leq a-1;\\
&\delta_j\leq a+1,& &\lambda_{1}+1\leq j\leq n. \\
\end{aligned} \right.
\end{align*}
According to Algorithm 2, we have
\begin{itemize}
 \item [$\bullet$] For $1\leq j\leq\lambda_{a+1}$, $T_j=\emptyset;$
 \item [$\bullet$] For $j=\lambda_{a+1}+\ell$ and $1\leq\ell\leq a-t_0$,
 $$T_{j}=\{k\!-\!\ell\!+\!1,k\!-\!\ell\!+\!2,\cdots\!,k\};$$
 \item [$\bullet$] For $j=\lambda_{a+1}+\ell$ and $a-t_0+1\leq\ell\leq a-1$
 $$T_{j}\!=\!\{k\!-\!\ell,k\!-\!\ell\!+\!1, \cdots\!,k\};$$
 \item [$\bullet$] Finally, $T_{\lambda_{1}+1}, T_{\lambda_{1}+2},\cdots\!,T_n$
 are obtained by dividing the sequence
 $$\underbrace{k\!-\!(a\!-\!t_0),\cdots\!,k},\underbrace{k\!-\!a,
 \cdots\!,k},\cdots,\underbrace{2,\cdots\!,k}$$ into $n\!\!-\!\!\lambda_{1}$
 segments of length $a$, and $T_{\lambda_{1}+j}$ is formed by the
 elements of the $j$th segment, $1\!\leq\! j\!\leq\! n\!-\!\lambda_{1}$.
\end{itemize}

For all of these subcases, similar to Case 1.1, it can be verified
that for each $j\in[n]$, $T_j$ is a subset of $[k]$, $S_j\cap
T_j=\emptyset$ and, when viewed as multisets, we have
$\sqcup_{j=1}^nT_j=T$.

\vspace{0.1cm}Next, consider Case 2. We need to divide this case
into the following three subcases according to the value of $r$.

\textbf{Case 2.1}: $r\leq \lambda_{a+1}$.

Then by \eqref{def-u}, \eqref{def-v} and \eqref{S-size-v-2}, we
have
\begin{align*}
\theta_{j}=\left\{\begin{aligned}
&0,& &1\leq j\leq \lambda_{a+1}; \\
&\ell-1,& &j=\lambda_{a+1}+\ell\text{~and~}1\leq\ell\leq a;\\
&\delta_j=a,& &\lambda_{1}+1\leq j\leq n. \\
\end{aligned} \right.
\end{align*}
According to Algorithm 2, we have
\begin{itemize}
 \item [$\bullet$] For $1\leq j\leq\lambda_{a+1}+1$, $T_j=\emptyset;$
 \item [$\bullet$] For $j=\lambda_{a+1}+\ell$ and $2\leq\ell\leq a$,
 $$T_{j}=\{k\!-\!\ell\!+\!2,k\!-\!\ell\!+\!3,\cdots,k\};$$
 \item [$\bullet$] Finally, $T_{\lambda_{1}+1}, T_{\lambda_{1}+2},\cdots,T_n$
 are obtained by dividing the sequence
 $$\underbrace{k\!-\!a\!+\!1, \cdots\!,k},
 \cdots,\underbrace{2,\cdots\!,k}$$ into $n\!\!-\!\!\lambda_{1}$
 segments of length $a$, and $T_{\lambda_{1}+j}$ is formed by the
 elements of the $j$th segment, $1\!\leq\! j\!\leq\!
 n\!-\!\lambda_{1}$.
\end{itemize}

\textbf{Case 2.2}: $\lambda_{a+1}<r\leq \lambda_{1}$.

Then $r=\lambda_{a+1}+t_1$, where $1\leq t_1\leq a$ and
\begin{align*}
\theta_{j}=\left\{\begin{aligned}
&0,& &1\leq j\leq \lambda_{a+1}; \\
&\ell,& &j=\lambda_{a+1}+\ell\text{~and~}1\leq\ell\leq t_1;\\
&\ell-1,& &j=\lambda_{a+1}+\ell\text{~and~}t_1+1\leq\ell\leq a;\\
&\delta_j=a,& &\lambda_{1}+1\leq j\leq n. \\
\end{aligned} \right.
\end{align*}
According to Algorithm 2, we have
\begin{itemize}
 \item [$\bullet$] For $1\leq j\leq\lambda_{a+1}$, $T_j=\emptyset;$
 \item [$\bullet$] For $j=\lambda_{a+1}+\ell$ and $1\leq\ell\leq t_1$,
 $$T_{j}=\{k\!-\!\ell\!+\!1,k\!-\!\ell\!+\!2,\cdots,k\};$$
 \item [$\bullet$] For $j=\lambda_{a+1}+\ell$ and $t_1+1\leq\ell\leq a$,
 $$T_{j}=\{k\!-\!\ell\!+\!2,k\!-\!\ell\!+\!3,\cdots,k\}
 \backslash\{k\!-\!\ell\!+\!t_1\!+\!1\};$$
 \item [$\bullet$] Finally, $T_{\lambda_{1}+1}, T_{\lambda_{1}+2},\cdots,T_n$
 are obtained by dividing the sequence
 $$\underbrace{k\!-\!a\!+\!t_1\!+\!1, \cdots\!,k},
 \underbrace{k\!-\!a\!+\!1, \cdots\!,k},
 \cdots,\underbrace{2,\cdots\!,k}$$ into $n\!\!-\!\!\lambda_{1}$
 segments of length $a$, and $T_{\lambda_{1}+j}$ is formed by the
 elements of the $j$th segment, $1\!\leq\! j\!\leq\!
 n\!-\!\lambda_{1}$.
\end{itemize}

\textbf{Case 2.3}: $\lambda_{1}<r<n$.

Then by \eqref{def-u}, \eqref{def-v} and \eqref{S-size-v-2}, we
have
\begin{align*}
\theta_{j}=\left\{\begin{aligned}
&0,& &1\leq j\leq \lambda_{a+1}; \\
&\ell,& &j=\lambda_{a+1}+\ell\text{~and~}1\leq\ell\leq a;\\
&\delta_j\leq a+1,& &\lambda_{1}+1\leq j\leq n. \\
\end{aligned} \right.
\end{align*}
According to Algorithm 2, we have
\begin{itemize}
 \item [$\bullet$] For $1\leq j\leq\lambda_{a+1}$, $T_j=\emptyset;$
 \item [$\bullet$] For $j=\lambda_{a+1}+\ell$ and $1\leq\ell\leq a$,
 $$T_{j}=\{k\!-\!\ell\!+\!1,k\!-\!\ell\!+\!2,\cdots,k\};$$
 \item [$\bullet$] Finally, $T_{\lambda_{1}+1}, T_{\lambda_{1}+2},\cdots,T_n$
 are obtained by dividing the sequence
 $$\underbrace{k\!-\!a, \cdots\!,k},
 \underbrace{k\!-\!a\!-\!1, \cdots\!,k},
 \cdots,\underbrace{2,\cdots\!,k}$$ into $n\!\!-\!\!\lambda_{1}$
 segments of length $a$, and $T_{\lambda_{1}+j}$ is formed by the
 elements of the $j$th segment, $1\!\leq\! j\!\leq\!
 n\!-\!\lambda_{1}$.
\end{itemize}

For all of these subcases, similar to Case 1.1, it can be verified
that for each $j\in[n]$, $T_j$ is a subset of $[k]$, $S_j\cap
T_j=\emptyset$ and, when viewed as multisets, we have
$\sqcup_{j=1}^nT_j=T$.

Combining all of the above discussions, we proved that each $T_j$
is a subset of $[k]$, $S_j\cap T_j=\emptyset$ and, when viewed as
multisets, $\sqcup_{j=1}^nT_j=T$.

\section{Proof of Claim 3}
We again consider all the cases and subcases discussed in
Appendices A and B.

We use the notations $\lambda_1$ and $\lambda_{a+1}$ with the same
meaning as in Appendices A and B. We further define $\lambda_j$
for all $j\in\{2,\cdots,a\}\cup\{a+2,\cdots,k-1\}$ as follows.

For Case 1, let
\begin{align}\label{def-lmd-rm-1}
\lambda_{j}=\left\{\begin{aligned}
&\lambda_{a+1}+a-j+1,& &\text{if} ~t_0+1\leq j\leq a; \\
&\lambda_{a+1}+a-j,& &\text{if} ~2\leq j\leq t_0.\\
\end{aligned} \right.
\end{align} And for Case 2, let
\begin{align}\label{def-lmd-rm-2}
\lambda_{j}=\lambda_{a+1}+a-j+1, ~\forall ~2\leq j\leq a.
\end{align}

For each $j\in\{a+2,\cdots,k-1\}$, let $\lambda_j\in[n]$ be such
that
\begin{align}\label{lmd-j}
\sum_{\ell=1}^{\lambda_{j}-1}\delta_\ell<\sum_{\ell=j}^{k-1}\ell\leq
\sum_{\ell=1}^{\lambda_{j}}\delta_\ell.\end{align} Note that by
\eqref{def-u} and \eqref{avg}, we have $\delta_j\!\leq\!
a+1\!\leq\! k-1$ for each $j\in[n]$. Then for each
$j\in\{a+2,\cdots,k-1\}$, it is easy to see that $\lambda_j$ is a
uniquely determined value. 

As an illustration, consider again Example \ref{exm-cnstr-W-2}.
Note that $k=7$ and $a=4$, and in Appendix A, we have obtained
$\lambda_{5}=3$ and $\lambda_{1}=6$. Now we can further obtain
$\lambda_{6}=2$, $\lambda_{3}=\lambda_{4}=4$ and $\lambda_{2}=5$.
In general, for Case 1, by \eqref{def-lmd-rm-1} and \eqref{lmd-j},
we always have
$$\lambda_{k-1}<\cdots<\lambda_{t_0-1}<\lambda_{t_0}
=\lambda_{t_0+1}<\cdots<\lambda_{1}.$$

For Example \ref{exm-cnstr-W-1}, note that $k=7$ and $a=3$, and in
Appendix A, we have obtained $\lambda_{4}=4$ and $\lambda_{1}=7$.
We can further obtain $\lambda_{6}=2$, $\lambda_{5}=3$,
$\lambda_{3}=5$ and $\lambda_{2}=6$. In general, for Case 2, by
\eqref{def-lmd-rm-2} and \eqref{lmd-j}, we always have
$$\lambda_{k-1}<\lambda_{k-2}<\cdots<\lambda_{2}<\lambda_{1}.$$

Now let
$X^*=\{(1,|S_1|),(2,|S_2|),\cdots,(\lambda_1,|S_{\lambda_1}|)\}$
and suppose $X^*$ is represented as $X^*=X_1^*\sqcup
X_2^*\sqcup\cdots\sqcup X_k^*$ for some $\sigma^*\in\mathscr S_k$
and some $(X_1^*,X_2^*,\cdots,X_k^*)\in\mathcal X_{\sigma^*}$.
Then for all subcases as discussed in the proof of Claim 2, it is
a mechanical (but somewhat tedious) work to check, just as in
Example \ref{exm-cnstr-W-2}, that
$$X_k^*=\emptyset$$ and
$$X_{i}^*=Z_{i}\cap\{1,2,\cdots,\lambda_{i}\}, ~\forall
i\in\{1,2,\cdots,k-1\}.$$ Hence, $\sigma^*: i\mapsto k-i+1,
~\forall i\in[k]$, is the unique permutation in $\mathscr S_k$ and
$(X^*_1,X^*_2,\cdots,X^*_k)$ is the unique choice in $\mathcal
X_{\sigma^*}$.

\end{document}